\documentclass[cmp]{svjourmod}

\usepackage{dcolumn}% Align table columns on decimal point
\usepackage{bm}% bold math
\usepackage{verbatim}       % Defines \begin{comment}, \end{comment}

\usepackage[dvips]{graphicx}
\usepackage{amssymb}
\usepackage{bm}
\usepackage{psfrag}
\usepackage{amsmath}
\usepackage{amsfonts}
\newcommand{\dd}{{\rm d}}
\begin{document}
%
%\DeclareGraphicsExtensions{.pdf}

\title{Chronological spacetimes without lightlike lines are stably causal}

%\author{E. Minguzzi \footnote{Dipartimento di Matematica Applicata, Universit\`a degli Studi di Firenze,  Via
%S. Marta 3,  I-50139 Firenze, Italy. E-mail:
%ettore.minguzzi@unifi.it}}

\author{E. Minguzzi}
\institute{Dipartimento di Matematica Applicata, Universit\`a degli
Studi di Firenze,  Via S. Marta 3,  I-50139 Firenze, Italy \\
\email{ettore.minguzzi@unifi.it} }
% \\ Phone: +39 055 4796 253, Fax: +39 055 471787 }
\authorrunning{E. Minguzzi}

\date{}
\maketitle

\begin{abstract}
\noindent The statement of the title is proved. It implies that
under physically reasonable  conditions,  spacetimes which are free
from singularities are necessarily stably causal and hence admit a
time function. Read as a singularity theorem it states that if there
is some form of causality violation on spacetime then either it is
the worst possible, namely violation of chronology, or there is a
singularity. The analogous result: ``Non-totally vicious spacetimes
without lightlike rays are globally hyperbolic'' is also proved, and
 its physical consequences are explored.
\end{abstract}

%\pacs{}

%\noindent Key Words:

\section{Introduction}
While the local structure of spacetime is fairly simple to describe,
there are still  a number of open problems concerning the causal
behavior of the spacetime manifold in the large. About three decades
ago Geroch and Horowitz in the conclusions of their review ``Global
structure of spacetimes'' \cite{geroch79} identified the problem of
giving good physical reasons for assuming stable causality as one of
the most important questions concerning the global aspects of
general relativity together with the proof of the cosmic censorship
conjecture. Indeed, if stable causality holds then the spacetime
does not suffer any pathological behavior connected with the
presence of almost closed causal curves, and, more importantly, it
admits a  (non-unique) time function \cite{hawking68}, that is a
function which is continuous and increases on every causal curve.

In order to understand the role of stable causality it is useful to
recall that most conformal invariant properties  can be ordered in
the so called causal ladder of spacetimes (see figure \ref{layer}).
If the real Universe were represented by a globally hyperbolic
manifold (the top of the ladder) then a number of mathematical and
physical nice properties would hold. The problem is that, though
there is evidence that the spacetime manifold evolves according to
the Einstein equations, it is not clear whether the evolution from
physically reasonable Cauchy data would introduce naked
singularities and would eventually produce a non-globally hyperbolic
spacetime. If so, the Cauchy data would be insufficient for the
determination of the spacetime geometry and one would have to take
into account the information coming from infinity. However, Penrose
gave arguments which support the view that the so developed manifold
would actually be globally hyperbolic \cite{penrose79} (strong
cosmic censorship).

Some other authors claim that one should only expect that the
non-predictable behavior due to naked singularities be confined
behind horizons (weak cosmic censorship). Other authors note that
there is not even compelling reasons for excluding chronologically
violating regions, in fact in some cases they allow to keep the
spacetime non-singular even in presence of trapped surfaces
\cite{newman89}. From this point of view chronology violating sets
should not be discarded a priori, instead they should be considered
in the same footing as naked singularities, a physical possibility
which hopefully remains hidden behind an horizon. These
considerations show that  the class of mathematically reasonable
spacetimes is considerably large, and therefore physicists look for
physical arguments which allow to get as close as possible to global
hyperbolicity. In short physicists look for results which allow to
climb the causal ladder.

The first step would be to justify the chronology property. Actually
this assumption is philosophically satisfactory  because its
violation would arise issues related to the free will of the generic
observer. However,  the notion of free will is not modeled in
general relativity, therefore it becomes reasonable to search for
other physical mechanisms, perhaps based on quantum mechanics, which
prevent the formation or stability of chronology violating sets. The
idea that such  a mechanism should indeed exist and that  starting
from well behaved initial conditions closed timelike curves can not
form has been referred by Hawking as the {\em chronology protection
conjecture} \cite{hawking92}. As I commented above there is no
general consensus on its validity and the  evidence coming from
classical general relativity is  under investigation
\cite{tipler77,thorne93,visser96,krasnikov02}.

It is natural to separate the remainder of the causal ladder in two
parts. That going from chronology up to stable causality (causality,
distinction, strong causality belong to it), and that going from
stable causality up to global hyperbolicity (passing through causal
continuity and causal simplicity). While the former part deals with
each time more demanding conditions conceived  to avoid almost
closed causal curves, the latter part presents each time more
demanding conditions in order to reduce the effects of points at
infinity on spacetime.

The problem of climbing the causal ladder from chronology up to
stable causality will be considered and solved in this work. It has
received less attention than the latter problem, that is, that of
going from stable causality up to global hyperbolicity which is
indeed more closely related to the { strong cosmic censorship
conjecture} \cite{penrose79}.

I am going to prove that {\em chronology plus the absence of
lightlike lines implies stable causality} (theorem \ref{vge}). The
theorem is formulated so that every mentioned property is
conformally invariant. It is therefore a theorem on the causal
structure of spacetime. In this respect it is important to use the
weaker assumption of {\em absence of lightlike lines} instead of the
more common null convergence, genericity and completeness
conditions, though these have a more direct physical meaning. In any
case the requirement of absence of lightlike lines  can be regarded
as a null completeness assumption, that is,  it follows from
demanding absence of singularities. I shall say more on this
correspondence in the first section. Thus the theorem physically can
be interpreted by saying that {\em under chronology, the absence of
singularities implies stable causality and hence the existence of a
time function}. It is the first result of this form which reduces
the existence of a time function to considerable less demanding
properties. Moreover,  note that in the previous statement the
required absence of singularities  is more precisely only a null
completeness requirement: the spacetime manifold could still be
timelike incomplete in a way compatible with the singularity
theorems (I shall say more on that in sections \ref{boj} and
\ref{ojb}).

Recall that stable causality is the best possible constraint in
order to remove almost closed causal curves and hence causality
violations. The theorem can then be regarded as a singularity
theorem, indeed, rewritten in the form {\em non-stably causal
spacetimes either are non-chronological or admit  lightlike lines}
receives the following physical interpretation {\em if there is a
form of causality violation on spacetime then either it is the worst
possible, namely violation of chronology, or the spacetime is
singular}. Regarded in this way the theorem clarifies the influence
of causality violations on singularities. In fact, if the violation
of chronology is regarded as a sort of singularity then the theorem
states that if there is no time function then the spacetime is
singular in this broader sense.

I refer the reader to \cite{minguzzi06c,minguzzi07b} for most of the
conventions used in this work. In particular, I denote with $(M,g)$
a $C^{r}$ spacetime (connected, time-oriented Lorentzian manifold),
$r\in \{3, \dots, \infty\}$ of arbitrary dimension $n\geq 2$ and
signature $(-,+,\dots,+)$. On $M\times M$ the usual product topology
is defined. For convenience and generality I often use the causal
relations on $M \times M$ in place of the more widespread point
based relations $I^{+}(x)$, $J^{+}(x)$, $E^{+}(x)$ (and past
versions). All the causal curves that we shall consider are future
directed (thus also the past rays). The subset symbol $\subset$ is
reflexive, $X \subset X$. Several versions of the limit curve
theorem will be repeatedly used, particularly those referring to
sequences of $g_n$-causal curves, where the metrics in the sequence
$g_n$ may differ. The reader is referred to \cite{minguzzi07c} for a
sufficiently strong formulation.

\section{Absence of lightlike lines}

In this section I consider the property of {\em absence of lightlike
lines} and comment on its physical meaning.

Two spacetimes belonging to the same conformal class $(M, \bm{g})$
share the same lightlike geodesics up to reparametrizations, and the
condition of maximality for the lightlike geodesic $\gamma$ reads
``there is no pair of events $x,z \in \gamma$, $(x,z) \in I^+$'',
which makes no mention to the full metric structure and hence is
independent of the representative of the conformal class. Thus, it
is convenient to give the following conformal invariant definition,

\begin{definition}
A lightlike line is an achronal inextendible causal curve.
\end{definition}

The definition implies, by achronality, that the causal curve is a
lightlike geodesic and that it  maximizes the Lorentzian length
between any of its points.

It is well known that \cite[Chap. 10, Prop. 48]{oneill83}

\begin{proposition}
If a inextendible lightlike geodesic admits a pair of conjugate
events then it is not a lightlike line.
\end{proposition}
 It can be proved that the notion of conjugate
points along a lightlike geodesic is conformally invariant
\cite{minguzzi06c}, thus the previous proposition relates two
conformal invariant properties. In particular note that the
requirement {\em every lightlike geodesic has a pair of conjugate
points} is stronger than {\em absence of lightlike lines}, e.g. 1+1
Minkowski spacetime with $x=0$ and $x=1$ identified. From the point
of view of Lorentzian geometry any statement should be formulated so
as to make its conformal invariance clear. For physical reasons some
authors prefer to mention physically motivated but non-conformal
invariant conditions. The consequence, however, is that several
results have been formulated in an unnecessary weak form as the
assumptions of the theorems are not really used.

\begin{definition}
An inextendible   lightlike geodesic $\gamma$  of the spacetime
$(M,g)$ satisfies the {\em generic condition} if at some $x \in
\gamma$ the tangent vector $n$ to the curve  is a {\em generic
vector}, that is, $n^{c} n^{d} n_{[a} R_{b] c d [e} n_{f]} \ne 0$. A
spacetime satisfies the {\em null generic condition} if every
inextendible lightlike geodesic satisfies the generic condition.
\end{definition}

A spacetime can be generic only if $n\ge 3$ (see \cite[Cor.
2.10]{beem96}). The precise sense in which the {\em null generic
condition} is generic is clarified by \cite[Prop. 2.15]{beem96}. It
is usually assumed on the physical ground that if a lightlike
geodesic does not satisfy it then arbitrarily small metric
perturbation in the geodesic path would make it true.

\begin{definition}
The spacetime $(M,g)$ satisfies the timelike convergence condition
if  $R(v,v) \ge 0$ for all timelike, and hence also for all
lightlike, vectors $v$. The spacetime $(M,g)$ satisfies the null
convergence condition if  $R(v,v) \ge 0$ for all lightlike vectors
$v$ (cf. \cite[p.95]{hawking73} \cite[Def. 12.8]{beem96}).
\end{definition}

Thus the null convergence condition is  a consequence of the
positivity of the energy density.

\begin{definition}
A spacetime $(M,g)$ is null geodesically complete if every
inextendible lightlike geodesic is complete.
\end{definition}

\begin{proposition} \label{mfe}
In a spacetime $(M,g)$  of dimension $\dim M \ge 3$, which satisfies
the null convergence condition, the null generic condition and that
is null geodesically complete every inextendible lightlike geodesic
admits a pair of conjugate events.  In particular $(M,g)$ does not
have lightlike lines.
\end{proposition}
\begin{proof}
It follows from the existence of some pair of conjugate points in
the lightlike geodesics accordingly to \cite[Prop. 4.4.5]{hawking73}
\cite[Prop. 12.17]{beem96}.
\end{proof}
This proposition has been improved by Tipler
\cite{tipler78,tipler78b} and Chicone and Ehrlich \cite{chicone80}
(see also Borde \cite{borde87}) by weakening the null convergence
condition to the  averaged null convergence condition. This
possibility is important because many quantum fields on spacetime
determine a stress-energy tensor and hence a Ricci tensor which does
not comply with the null convergence condition while it satisfies
the averaged null convergence condition.

Proposition \ref{mfe} implies that the condition of {\em absence of
lightlike lines} is quite reasonable from a physical point of view
at least if the spacetime is assumed to be non-singular (see also
the discussion in \cite[Sect. 4.4]{hawking73}) or just null
geodesically complete.

In the next sections I will prove that the assumption of {\em
absence of lightlike lines} has the effect of identifying the levels
of the causal ladder between chronology and stable causality. In
this respect the hard part will come with the inclusion of stable
causality. A key role will be played by the property of
$K$-causality introduced by Sorkin and Woolgar \cite{sorkin96}, and
for the last step by a new property which I study in the next
section.

%
%
%
%\begin{figure}[ht]
%\centering \psfrag{A}{ {\footnotesize \fbox{
%\begin{tabular}{c}
%\vspace{-0.3cm}\\
%Global hyperbolicity \\
%$\Downarrow $ \\
%Causal simplicity \\
%$\Downarrow $ \\
%Causal continuity \\
%$\Downarrow $ \\
%{Stable causality} \\
%$\quad \ \downarrow \ \uparrow\, ? $ \\
%{$K$-causality} \\
%$\Downarrow$ \\
%{$\overline{A^{\infty}}$-causality} \\
%$\Downarrow $ \\
%{Compact stable causality} \\
%$\Downarrow $ \\
%{$A^{\infty}$-causality} \\
%$\Downarrow $ \\
%$A$-causality\\
%$\Downarrow $ \\
%Strong causality \\
%$\Downarrow $ \\
%Non-partial imprisonment \\
%$\Downarrow $ \\
%Distinction \\
%$\Downarrow $ \\
%Non-total imprisonment \\
%$\Downarrow $ \\
%Causality\\
%$\Downarrow $ \\
%Chronology \\
%$\Downarrow $ \\
%Non-total viciousness \\
%\end{tabular}}
%} }
%\includegraphics[height=287pt]{layer}
%\caption{The  causal  ladder displaying the new levels considered in
%section \ref{vgq}. Penrose's infinite ladder between $A$-causality
%and $A^{\infty}$-causality is omitted \cite{minguzzi07b}, as well as
%the levels of weak distinction and feeble distinction
%\cite{minguzzi07e}. For the placement of the non-imprisonment
%properties the reader is referred to \cite{minguzzi07f}. The arrow
%$C \Rightarrow D$ means that $C$ implies $D$ and there are examples
%which show that $C$ differs from $D$. Stable causality implies
%$K$-causality but it is not known if they coincide. The implications
%climbing the ladder express the geometrical content of the theorems
%proved in this work.} \label{layer}
%\end{figure}
%
%

\begin{figure}[ht]
\centering
\includegraphics[height=287pt]{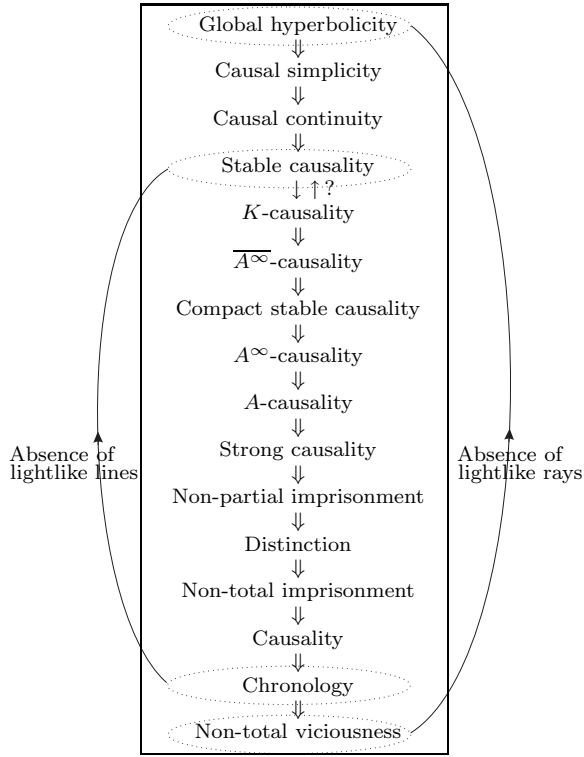}
\caption{The  causal  ladder displaying the new levels considered in
section \ref{vgq}. Penrose's infinite ladder between $A$-causality
and $A^{\infty}$-causality is omitted \cite{minguzzi07b}, as well as
the levels of weak distinction and feeble distinction
\cite{minguzzi07e}. For the placement of the non-imprisonment
properties the reader is referred to \cite{minguzzi07f}. The arrow
$C \Rightarrow D$ means that $C$ implies $D$ and there are examples
which show that $C$ differs from $D$. Stable causality implies
$K$-causality but it is not known if they coincide. The implications
climbing the ladder express the geometrical content of the theorems
proved in this work.} \label{layer}
\end{figure}

\clearpage

\section{Compact stable causality} \label{vgq}

Recall that a non-total imprisoning spacetime is a spacetime for
which there is no future-inextendible causal curve totally
imprisoned in a compact (future non-total imprisonment is equivalent
to past non-total imprisonment \cite{beem76,minguzzi07f}). It is
known that every relatively compact open set  in a non-total
imprisoning spacetime \cite{minguzzi07f} is stably causal when
regarded as a spacetime with the induced metric \cite{beem76}.
Actually, this property characterizes non-total imprisonment, indeed
we have
\begin{theorem}
A spacetime $(M,g)$ is non-total imprisoning iff for every
relatively compact open set  $B$, $(B,g\vert_{B})$ is stably causal.
\end{theorem}

\begin{proof}
The implication to the right was proved by Beem \cite{beem76}. To
the left, assume $(M,g)$ has a compact $C$ in which some curve
$\gamma$ is future imprisoned. In \cite{minguzzi07f} I proved that
there is a lightlike line $\eta$ contained in $C$ such that
$\eta\subset \Omega_{f}(\eta)$ where $\Omega_f(\eta)$ is the set of
accumulation points in the future of $\eta$ (in analogy with the set
of $\omega$-limit points of dynamical systems). Let $B$ be a
relatively compact open set  such that $C\subset B$. Take $q \in
\eta$ and, given a convex neighborhood $U\ni q$, $U\subset B$, take
$p \in \eta \cap J^{-}_{(U,g\vert_U)}(q)$. Take $g'>g$ in $B$ ($g'$
need not be defined on $B^C$) then $p \in I^{-}_{(U,g')}(q)$, but
recall that $p \in \Omega_f(\eta)$ is an accumulation point for the
future-inextendible $g'$-timelike curve given by the portion of
$\eta$ which starts from $q$. Thus since $I^{-}_{(U,g'\vert_U)}$ is
open it is possible to construct a closed $g'$-timelike curve
contained in $B$. The argument holds for any choice of $g'$ thus it
is not true that for every relatively compact open set  $B$,
$(B,g\vert_{B})$ is stably causal.
%
%
%This reasoning holds whatever the choice of $g'$ made above. Now
%note that $(B,g\vert_{B})$ cannot be stably causal because there
%would be $g''>g$ on $B$  without closed timelike curves, but then it
%would be possible to construct $g'=\chi g''+(1-\chi) g$ with $\chi:M
%\to [0,1]$ that vanishes (only) on $B^C$ so that $g'$ is such that
%$g'>g$ in $B$ and $g'=g$ in $B^{C}$ and being $g'\le g''$, $(M,g')$
%does not have closed timelike cures a contradiction because all the
\end{proof}

Note that non-total imprisonment is a quite weak property (it is
implied by weak distinction \cite{minguzzi07f}). A related problem
is that of establishing if, given an arbitrary  compact on
spacetime, the metric can be widened in it without introducing
closed causal curves in the {\em whole} spacetime. If this is
possible the spacetime satisfies a  condition which is stronger than
non-total imprisonment. We can define a new property

\begin{definition}
A spacetime $(M,g)$ is {\em  compactly stably causal} if for every
relatively compact open set $B$ there is a metric $g_B\ge g$ such
that $g_B>g$ on $B$, $g_B=g$ on $B^C$ and $(M,g_B)$ is causal.
\end{definition}

\begin{remark}
There are some equivalent definitions, for instance: $(M,g)$ is {\em
compactly stably causal} if for every compact set $C$ there is
$g_C\ge g$ such that $g_C>g$ on $C$ and $(M,g_C)$ is causal. In
order to prove the equivalence one has to take appropriate convex
combinations of metrics with smooth coefficients.
\end{remark}
%
%\begin{itemize}
%\item[(i)]
%\item[(ii)] for every relatively compact set $B$ there is a metric $g_B\ge
%g$ such that $g_B>g$ on $B$ and $(M,g_B)$ is causal.
%\item[(iii)] for
%every compact $C$ there is $g_C\ge g$ such that $g_C>g$ on $C$ and
%$(M,g_C)$ is causal.
%\end{itemize}
%\end{definition}
%
%\begin{proof}[Proof of the equivalence.] (i) $\Rightarrow$ (ii). Take $g_B=\tilde{g}_B$. (ii) $\Rightarrow$ (i).
%Take a convex combination of $g_B$ with $g$, $\tilde{g}_B=\chi
%g_B+(1-\chi)g$ where $\chi \in [0,1]$ is a function which is
%positive in $B$ and vanishes outside $B$. Since $\tilde{g}_B\le
%g_B$, $(M,\tilde{g}_B)$ is causal. (ii) $\Rightarrow$ (iii). Take
%$B$ such that $C \subset B$, and define $g_C=g_B$. (iii)
%$\Rightarrow$ (ii). Take $C$ such that $\bar{B} \subset C$, and
%define $g_B=\chi g_C+(1-\chi)g$ where $\chi\in [0,1]$ is a function
%which is positive in $B$ and vanishes outside $B$. Since $g_B\le
%g_C$, $(M,g_B)$ is causal.
%\end{proof}

Some natural questions arise, among them the placement of compact
stable causality in the causal ladder of spacetimes. Before
considering this question let me recall some notation and
terminology \cite{minguzzi07b}. Following Woodhouse
\cite{woodhouse73,akolia81} I denote with $A^+$ the closure of the
causal relation, that is $A^{+}=\bar{J}^+$. A spacetime is
$A^\infty$-causal if there is no finite cyclic chain of distinct
$A^+$-related events. This property is equivalent to the
antisymmetry of the relation $A^{+\infty}=\cup^{+\infty}_{i=1}
(A^{+})^i$, which is the smallest transitive relation containing
$A^{+}$. Analogously, a spacetime is $\overline{A^{\infty}}$-causal
if the relation $\overline{A^{+\infty}}$ is antisymmetric. The
relation $K^{+}$ is the smallest closed and transitive relation
containing $J^{+}$, and the spacetime is $K$-causal if the relation
$K^{+}$ is antisymmetric \cite{sorkin96}. It is known that stable
causality implies $K$-causality, although it is not known if these
two conditions coincide \cite{minguzzi07}. We have

%
%\begin{figure}[ht]
%\begin{center}
%
%{\footnotesize \fbox{
%\begin{tabular}{c}
%\vspace{-0.3cm}\\
%Global hyperbolicity \\
%$\Downarrow $ \\
%Causal simplicity \\
%$\Downarrow $ \\
%Causal continuity \\
%$\Downarrow $ \\
%{Stable causality} \\
%$\quad \ \downarrow \ \uparrow\, ? $ \\
%{$K$-causality} \\
%$\Downarrow$ \\
%{$\overline{A^{\infty}}$-causality} \\
%$\Downarrow $ \\
%{Compact stable causality} \\
%$\Downarrow $ \\
%{$A^{\infty}$-causality} \\
%$\Downarrow $ \\
%$A$-causality\\
%$\Downarrow $ \\
%Strong causality \\
%$\Downarrow $ \\
%Non-partial imprisonment \\
%$\Downarrow $ \\
%Distinction \\
%$\Downarrow $ \\
%Non-total imprisonment \\
%$\Downarrow $ \\
%Causality\\
%$\Downarrow $ \\
%Chronology \\
%$\Downarrow $ \\
%Non-total viciousness \\
%\end{tabular}}
%}
%\end{center}
%\caption{The  causal  ladder displaying the new levels considered in
%theorem \ref{mjh}. Penrose's infinite ladder between $A$-causality
%and $A^{\infty}$-causality is omitted \cite{minguzzi07b}, as well as
%the levels of weak distinction and feeble distinction
%\cite{minguzzi07e}. For the placement of the non-imprisonment
%properties the reader is referred to \cite{minguzzi07f}. The arrow
%$C \Rightarrow D$ means that $C$ implies $D$ and there are examples
%which show that $C$ differs from $D$. It is not known if stable
%causality is equivalent to $K$-causality, it is only known that the
%former implies the latter.}
%\end{figure}

\begin{theorem} \label{mjh}
$K$-causality implies $\overline{A^{\infty}}$-causality.
 \end{theorem}

\begin{proof}
Since $J^{+}\subset K^{+}$, any causal relation obtained from
$J^{+}$ by taking closures or by making the relation transitive
through the replacement $R^{+} \to \cup^{+\infty}_{i=0} (R^{+})^i$,
is still contained in $K^{+}$. Since $\overline{A^{+\infty}}$ has
this form $\overline{A^{+\infty}}\subset K^{+}$,   thus
$K$-causality implies $\overline{A^{\infty}}$-causality.
\end{proof}

\begin{remark}
Given a relation $R^{+}$ the two involutive operations given by (a)
closure: $R^{+} \to \bar{R}^{+}$, and (b) transitivization: $R^{+}
\to R^{+\infty}=\cup^{+\infty}_{i=1} (R^{+})^i$, once alternatively
applied to $J^{+}$ generate a chain of relations all contained in
$K^{+}$ whose first members are $J^{+}$, $A^{+}$, $A^{+\infty}$,
$\overline{A^{+\infty}}, \cdots$. By demanding the antisymmetry one
obtains a ladder of causal properties whose first members are
causality, $A$-causality, $A^{\infty}$-causality and
$\overline{A^{\infty}}$-causality, all necessarily weaker than
$K$-causality. If at a certain point two adjacent relations coincide
then they coincide with $K^+$ as they are both closed an transitive
and they are certainly the smallest relations with this property. In
this case the mentioned ladder of relations finishes there where
this coincidence occurs. As we shall see, the mentioned first levels
are all different but it is not known if from some point on the
levels would start to coincide, that is, if after a finite number of
operations of closure and transitivization one would get $K^{+}$ and
$K$-causality. Examples support the view that this coincidence
occurs at a level which increases with the dimensionality of the
spacetime.

%My own feeling after working with some (counter)examples is that
%this coincidence does indeed occur at a stage related to the
%dimensionality $n$ of the spacetime. Denote with a tilde  the
%operation of closure followed by that of transitivization then my
%guess is $K^{+}=(J^+)^{(\sim)^{n-1}}$, for instance, in a
%1+1-dimensional spacetime ($n=2$), $K^{+}=A^{+}$, while in a
%2+1-dimensional spacetime $K^{+}=\overline{A^{+\infty}}$.

\end{remark}

\begin{lemma}
Let $\circ$ denote the composition of relations, then $J^{+}\circ
A^{+}\subset A^{+}$  and $A^{+}\circ J^{+}\subset A^{+}$.
\end{lemma}

\begin{proof}
Let us consider the former case, the latter being analogous. Let
$(x,y) \in J^{+}$ and $(y,z) \in A^{+}$, and let $\gamma_n$ be a
sequence of causal curves of endpoints $(y_n,z_n) \to (y,z)$. Take
$x_k\in I^{-}(x)$, $x_k \to x$, so that $x_k\ll y$ and for
sufficiently large $n$, $x_k\ll y_n\le z_n$, thus $(x_k,z_{n(k)})
\in I^{+}$ and in the limit $(x,z) \in A^+$.
\end{proof}

\begin{theorem}
$\overline{A^{\infty}}$-causality implies compact stable causality.
\end{theorem}

\begin{proof}
 Suppose
$(M,g)$ is $\overline{A^{\infty}}$-causal but non-compactly stably
causal, then there is a relatively compact open set  $B$ such that
for every $ g'\ge g$, $g'>g$ on $B$, $g'=g$ on $B^C$, $(M,g')$ is
not causal. Let $g_n$ be a sequence of metrics $g_n\ge g$, $g_n
>g$ on $B$, $g_n=g$ on $B^C$, $g_{n+1}\le g_n$, and $g_n \to g$ pointwisely on the appropriate
tensor bundle. For every choice of $n$, $(M,g_n)$ is not causal, and
since $(M,g)$ is causal there must be a closed $g_n$-causal curve
$\gamma_n$ intersecting $B$ (see figure \ref{proof}). Let $p^0_n\in
\gamma_n\cap B$ and parametrize the curves with respect to a
complete Riemannian metric $h$ so that $p^0_n=\gamma_n(0)$.

Assume an infinite number of $\gamma_n$ is entirely contained in
$\bar{B}$. Beem \cite{beem76} has shown that there would be a
inextendible $g$-causal limit curve contained in $\bar{B}$ in
contradiction with the non-total imprisoning property of the
spacetime (recall that $A$-causality implies distinction which
implies the non-total imprisoning property). Thus without loss of
generality we can assume that none of the $\gamma_n$ is entirely
contained in $\bar{B}$. We conclude that $\gamma_n$ intersects
$\dot{B}$ at least once to enter $B^{C}$. Without loss of generality
we can also assume that $p^0_n \to p^0 \in \bar{B}$.

Using again the limit curve argument, through $p^0$ there passes a
future inextendible (hence its $h$-length parameter has domain
$(-\infty,+\infty)$) $g$-causal curve $\gamma^0$ which can't pass
through $p^0$ twice as it would imply a violation of causality for
$(M,g)$. In particular since $(M,g)$ is non-partial imprisoning it
escapes $\bar{B}$ at a last point $q^0 \in \dot{B}$ to never reenter
$\bar{B}$. Let $\gamma^0_n$ be a subsequence of $\gamma_n$ which
converges to $\gamma^0$ uniformly on compact subsets and let $s^0$
be the value of the parameter such that $q^0=\gamma^0(s^0)$. Since
$\gamma^0_n(s^0+2) \to \gamma^0(s^0+2) \notin \bar{B}$ pass to a
subsequence denoted in the same way so that $\gamma^0_n(s^0+2)
\notin \bar{B}$. Let $(\bar{s}^0_n,t^1_n)\ni s^0+2 $ be the largest
open connected interval so that
$\gamma^0_n((\bar{s}^0_n,t^1_n))\subset (\bar{B})^C$. Define
$\bar{q}^0_n, p^1_n\in \dot{B}$ as
$\bar{q}^0_n=\gamma^0_n(\bar{s}^0_n)$ and $p^1_n=\gamma^0_n(t^1_n)$.
Let $p^1 \in \dot{B}$ be an accumulation point for $p^1_n$, without
loss of generality we can assume $p^1_n \to p^1$. Note that the
segment $\gamma^0_n\vert_{[\bar{s}^0_n,t_1^n]}$ is entirely
contained in $B^C$ and hence it is $g$-causal. Since $\bar{s}^0_n
\in [0,s^0+2]$, without loss of generality we can assume
$\bar{s}^0_n \to \bar{s}^0$ for some $\bar{s}^0$. Now, $\bar{s}^0\le
s^0$ indeed if $\bar{s}^0 >s^0$ then $\bar{q}^0_n\in \bar{B}$
converges to $\gamma^0(\bar{s}^0)$ a point that does not belong to
$\bar{B}$ which is impossible. In particular, it is possible to find
a sequence $s^0_n$, $ \bar{s}^0_n < s^0_n <s^0+2$, such that
$s^0_n\to s^0$. Then $q^0_n=\gamma^0_n(s^0_n) \notin \bar{B}$
converges to $q^0$ and the $g$-causal sequence of curves
$\gamma^0_n\vert_{[s^0_n,t_1^n]}$ has endpoints $(q^0_n,p^1_n) \in
J^{+}$ such that $(q^0_n,p^1_n) \to (q^0,p^1)$, i.e. $(q^0,p^1)\in
A^{+}$. Note that $(p^0,q^0)\in J^{+}$ as both points belong to
$\gamma^0$, hence $(p^0,p^1)\in A^{+}$.

\begin{figure}[ht]
\centering \psfrag{a}{{\scriptsize $\, \bar{p}^0_n$}}
\psfrag{b}{{\scriptsize $\!p^0\!=\!\gamma^0(0)$}}
\psfrag{c}{{\scriptsize$\!\!\!\!\!\!\!\gamma^0_n(0)\!=\!p^0_n$}}
\psfrag{d}{{\scriptsize$\!\!\gamma^0_n$}}
\psfrag{e}{{\scriptsize$\gamma^0$}}
\psfrag{f}{{\scriptsize$q^0\!=\!\gamma^0(s_0)$}}
\psfrag{g}{{\scriptsize $ \!\!\!\!\!\!
\gamma^0_n(\bar{s}^0_n\!)\!=\!\bar{q}^0_n$}}
\psfrag{h}{{\scriptsize$\!\!\!\!\!\!\!\!\gamma^0_n(s^0_n)\!=\!q^0_n$}}
\psfrag{i}{{\scriptsize $\ \, \gamma^0\!(s_0+2)$}}
\psfrag{l}{{\scriptsize $ \bar{p}^1_n$}} \psfrag{m}{{\scriptsize
$\!\gamma^1$}} \psfrag{n}{{\scriptsize
$\!\!p^1_n\!=\!\gamma^1_n\!(\!0\!)$}} \psfrag{o}{{\scriptsize
$p^1$}} \psfrag{p}{{\scriptsize $\!\!\! q^1$}}
\psfrag{q}{{\scriptsize $\bar{q}^1_n$}} \psfrag{r}{{\scriptsize
$q^1_n$ }} \psfrag{s}{{\scriptsize $p^2_n$}} \psfrag{t}{{\scriptsize
$p^2$}} \psfrag{u}{{\scriptsize \!\!\!\!\!\!\!\!\!\! $g_n$-causal}}
\psfrag{u2}{{\scriptsize \!\!\!\!\!\!\!\!\!\! closed}}
\psfrag{v}{{\scriptsize $g$-causal}} \psfrag{v2}{{\scriptsize
inextendible}}
 \psfrag{z}{{\scriptsize
\!\!\!\!\!\!\!\!\!\! $g$-causal}} \psfrag{w}{{\scriptsize $B$}}
\includegraphics[width=6cm]{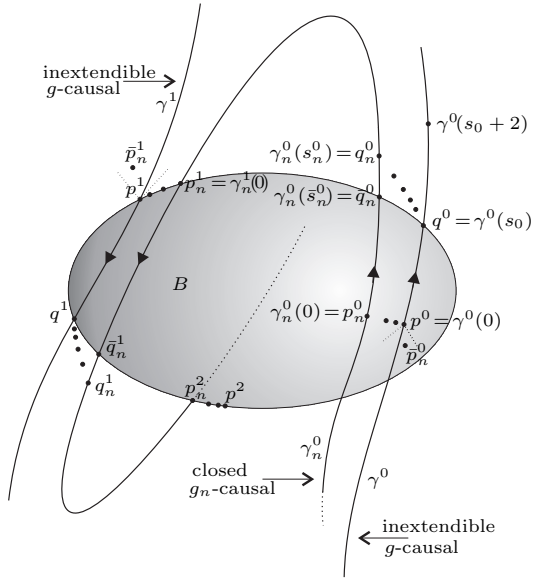}
\caption{The argument of the proof that
$\overline{A^{\infty}}$-causality implies compact stable causality.}
 \label{proof}
\end{figure}

The limit curve theorem states that $t^1_n \to +\infty$ otherwise
$p^1$ would belong to the prolongation of $\gamma^0$ which is
impossible since $q^0$ is the last point of $\gamma^0$ in $\bar{B}$.
The segments $\gamma^0_n\vert_{[\bar{s}^0_n,t_1^n]}$ are not all
contained in a compact because $\gamma^0$ escapes every compact to
never return and for every $k>0$, $\gamma_n^0(s_n^0+k) \to
\gamma^0(s^0+k)$ because $t^1_n \to +\infty$. As a consequence the
pair $(p^0,p^1)\in A^{+}$ can be regarded as the limit of the pairs
of endpoints of $g$-causal segments which are not all contained in a
compact. (In order to construct these segments take $\bar{p}^0_k \in
I^{-}(p^0)$, $\bar{p}^0_k \to p^0$ so that $q^0 \in
I^{+}(\bar{p}^0_k)$ and hence since $I^{+}$ is open $q^0_{n(k)} \in
I^{+}(\bar{p}^0_k)$ for a sufficiently large $n(k)$. Next follow the
$g$-causal segment $\gamma^0_n\vert_{[s^0_n,t_1^n]}$ which are not
all contained in a compact, finally redefine the parametrization of
the sequence $\bar{p}^0_k$ and pass if necessary to a subsequence so
that $(\bar{p}^0_n, q^0_n)\in I^{+}$ and hence $(\bar{p}^0_n,p^1_n)
\in J^{+}$ with $(\bar{p}^0_n,p^1_n) \to (p^0,p^1)$.) In particular,
$p^0\ne p^1$ since the spacetime is strongly causal.

Now, translate all the parametrizations of $\gamma^0_n$ so that
$t^1_n$ gets replaced by $0$. Repeat the previous steps where now
$p^1$ plays the role of $p^0$ and the found sequence $\gamma^1_n$ is
a reparametrized subsequence of $\gamma^0_n$.
%That is, construct the future inextendible curve $\gamma^3$
%passing through $p^2$, the uniformly converging sequence
%$\gamma^3_n$ (which is a reparametrized subsequence of
%$\gamma^1_n$), $\gamma^3_n(0) \to p^2$. Define $p^3$  to be the last
%point of $\gamma^3$ in $\bar{B}$ and define $p^4$ in a way
%completely analogous to $p^2$ so that $(p^3,p^4) \in \bar{A}$ and
%there are segments of $\gamma^3_n$ with endpoint converging to
%$(p^3,p^4)$, so that this  sequence of segments of $\gamma^3_n$ is
%not all contained in a compact.

Continue in this way, defining for at each step analogous
subsequences and events so that $p^k \in \bar{B}$, $(p^{k},p^{k+1})
\in A^{+}$, $p^{k}\ne p^{k+1}$, and for each $k$ there is a sequence
of $g$-causal curves, not all contained in a compact, so that the
endpoints of the sequence converge to $(p^{k},p^{k+1})$. Note that
for every pair of positive integers $a\le b$, $(p^a,p^b) \in
A^{+\infty}$.

Since $\bar{B} \times \bar{B}$ is compact, there is a subsequence
denoted $(p^{k_s},p^{k_s+1})$ such that $(p^{k_s},p^{k_s+1}) \to
(x,z)$ as $s \to +\infty$. Moreover, $x\ne z$ because otherwise for
every relatively compact causally convex neighborhood
 $U \ni x$, for sufficiently large $s$,
$(p^{k_s},p^{k_s+1}) \in U$, and the sequence of $g$-causal curves
not all contained in a compact, whose endpoints converge to
$(p^{k_s},p^{k_s+1})$ would contradict the causal convexity of $U$.
Since $A^{+}$ is closed, $(x,z) \in A^{+}$ and $x \ne z$. Since
$p^{k_s}$ is a subsequence of $p^k$, for every $s$, $k_{s}+1 \le
k_{s+1}$, thus  $(p^{k_{s}+1},p^{k_{s+1}}) \in A^{+\infty}$ and in
the limit $s \to +\infty$, $(z,x) \in \overline{A^{+\infty}}$. As a
consequence $(M,g)$ is not $\overline{A^{\infty}}$-causal which is
the searched contradiction.

\end{proof}

\begin{theorem}
Compact stable causality implies $A^{\infty}$-causality.
\end{theorem}

\begin{proof}
Assume the spacetime is compactly stably causal, and suppose it is
not $A^{\infty}$-causal then there is a finite closed chain of
$A^{+}$-related events $(x_i,x_{i+1}) \in A^{+}$, $i=1,\ldots, n$,
$x_{n+1}=x_1$.

Consider a relatively compact open set $B$ which contains all $x_i$,
$i=1,\ldots, n$,  and let $g_B\ge g$, $g_B> g$ on $B$, $g_B=g$ on
$B^C$. We want to prove that $A^{+}\cap(B\times B) \subset
J^{+}_{(M,g_B)}$, from which it follows that $(M,g_B)$ is not causal
whatever the choice of $g_B$ and hence $(M,g)$ is not compactly
stably causal, the searched contradiction. Let $(y,z) \in A^{+}$,
$y,z \in B$, then by the limit curve theorem either $(y,z) \in J^{+}
\subset J^{+}_{(M,g_B)}$ or there are a future inextendible
$g$-causal curve $\sigma^y$ starting from $y$, and a past
inextendible $g$-causal curve $\sigma^z$ ending at $z$ such that for
every $y' \in \sigma^y\backslash\{y\}$ and $z' \in
\sigma^z\backslash\{z\}$, $(y',z') \in A^{+}$. At least a segment of
$\sigma^y$ near $y$ is timelike for $(M,g_B)$ and analogously for
$\sigma^z$, thus $(y,y') \in  I^{+}_{(M,g_B)}$, and $(z',z) \in
I^{+}_{(M,g_B)}$ finally since $(y',z') \in A^{+} \subset
\overline{J^{+}_{(M,g_B)}}$, it is $(y,z) \in I^{+}_{(M,g_B)}$.

\end{proof}

\begin{remark} All the properties of the previous theorems differ.
In \cite{minguzzi07b} I gave an example of non-$K$-causal
$A^{\infty}$-causal spacetime. A closer inspection proves that it is
actually non-$\overline{A^{\infty}}$-causal but compactly stably
causal. Moreover, it is possible to construct an example, similar to
that of \cite{minguzzi07b}  which is $\overline{A^{\infty}}$-causal
but non-$K$-causal (simply repeat the figure of \cite{minguzzi07b}
three times vertically, and then identify the holes cyclically). The
properties $A^{\infty}$-causality and compact stable causality
differ because of the spacetime example of figure \ref{noncomp}. A
consequence of these examples is the perhaps surprising fact that
compact stable causality differs from stable causality (see again
the example of \cite{minguzzi07b}). This fact means that the
behavior of the light cones near infinity  is important in order to
determine if a spacetime is properly compactly stably causal or not.
\end{remark}

%This example is interesting because for every  enlarged metric $g_n$
%on $K$ there is, as already said, a $g_n$-causal  curve $\gamma_n$
%crossing $K$, but also as $g_n$ tends to $g$ the number of times
%$N(n)$ that $\gamma_n$ intersects $K$ goes  to infinity. This
%possible behavior was considered in the proof of theorem \ref{mjh}.

\begin{figure}[ht]
\centering \psfrag{K}{{\footnotesize $K$}} \psfrag{x}{{\footnotesize
$x$}} \psfrag{ay}{{\footnotesize $y$}}
\includegraphics[width=12cm]{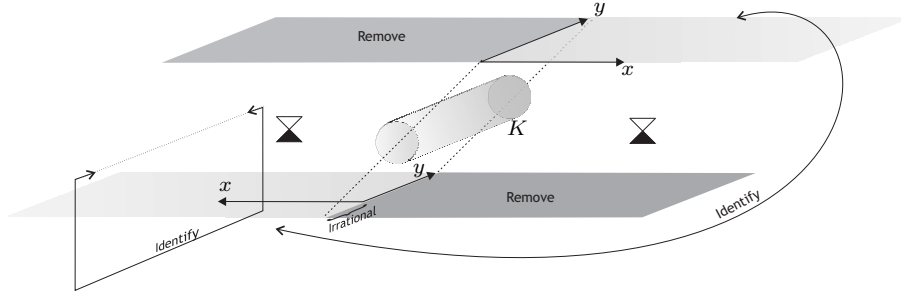}
\caption{A $A^{\infty}$-causal but non-compactly stably causal
spacetime. In order to construct the spacetime start from
$\mathbb{R} \times S^1 \times \mathbb{R}$ of coordinates
$(t,\theta,z)$, $\theta \in[0,1]$, and metric $g=-\dd t^2+\dd
\theta^2+\dd z^2$, remove two spacelike surfaces and identify, after
a translation by an irrational number,
 two spacelike surfaces as done in the figure . The coordinates
$(x,y)$ have been introduced on the identified surfaces so as to
make the identification clear. The spacetime is non-orientable but
this feature is not essential. The spacetime is non-compactly stably
causal since any enlargement of the metric on $K$ gives closed
causal curves. Thanks to the translation by an irrational number
there cannot be closed chains of $A^{+}$ related events.}
\label{noncomp}
\end{figure}

\section{The proof and some physical considerations} \label{boj}
I start with a result due to Hawking \cite{hawking70} \cite[Prop.
6.4.6]{hawking73} (he proved it with the stronger but inessential
assumption that every lightlike geodesic admits a pair of conjugate
points)

\begin{lemma} \label{vfd}
A chronological spacetime without lightlike lines is strongly
causal.
\end{lemma}

\begin{proof} Recall that in a strongly causal spacetime, given any
neighborhood $U$ of $x\in M$ there exist a neighborhood $V \subset
U$, $x \in V$, such that  any future-directed causal curve with
endpoints at $V$ is entirely contained in $U$ (see for instance
\cite[Lemma 3.22]{minguzzi06c}). Thus if $(M,g)$ were not strongly
causal there would be a point $x$, a neighborhood $U\ni x$, and a
sequence of causal curves $\gamma_n$ of starting event $x_n$, ending
event $z_n$ such that $x_n \to x$, $z_n \to x$, and the curves
$\gamma_n$ are not entirely contained in $U$. Hence there are the
conditions required by the limit curve theorem \cite[theorem
3.1]{minguzzi07c} case (2) which implies the existence of a
lightlike line passing through $x$, a contradiction. \end{proof}

%
%\begin{remark}
%The previous proof shows that under the assumption of {\em absence
%of lightlike lines} if strong causality is violated at $x$ then
%chronology is violated somewhere. If, however, one assumes the
%stronger property that every lightlike geodesic has a pair of
%conjugate points then the point $x$ must lie in the closure of the
%chronology violating set.
%\end{remark}

A fundamental step in the proof is
\begin{theorem}
If a spacetime does not have lightlike lines then the relation
$A^{+}=\bar{J}^{+}$ is transitive, that is $K^{+}=A^{+}$. Moreover,
if the spacetime is also chronological then the spacetime is
$K$-causal.
\end{theorem}

\begin{proof}

Let us prove the transitivity of $A^{+}$. Take two pairs $(x,y) \in
A^{+}$ and $(y,z) \in A^{+}$ and two sequences of causal curves
$\sigma_n$ of endpoints $(x_n, y_n) \to (x,y)$, and $\gamma_n$ of
endpoints $(y'_n,z_n) \to (y,z)$. Apply the limit curve theorem
\cite{minguzzi07c} to both sequences, and consider first the case in
which the limit curve in both cases does not connect the limit
points. By the limit curve theorem, $\sigma_n$ has a limit curve
$\sigma$ which is a past inextendible causal curve ending at $y$.
Analogously $\gamma_n$ has a limit curve $\gamma$ which is   a
future inextendible causal curve starting from $y$. The inextendible
curve $\gamma\circ\sigma$ cannot be a lightlike line thus there are
points $x'\in \sigma\backslash\{y\}$, $z' \in \gamma\backslash\{y\}$
such that $(x',z')\in I^{+}$ and (pass to a subsequence) points
$x'_n \in \sigma_n$, $x'_n \to x'$ and $z'_n \in \gamma_n$, $z'_n
\to z'$, thus, since $I^{+}$ is open, for sufficiently large $n$,
$(x_n,z_n) \in I^{+}$ and finally $(x,z) \in \bar{I}^{+}=A^{+}$.

If both limit curves join the limit points then clearly $(x,z) \in
J^{+}\subset A^{+}$. If, say, $\sigma$ joins $x$ to $y$ but $\gamma$
does not join $y$ to $z$, take $x'_n\in I^{-}(x)$, $x'_n \to x$, so
that $x'_n\ll y$ and for large $n$, $x'_n\ll y'_n\le z_n$, thus in
the limit $(x,z) \in A^+$. The remaining case is analogous. Thus
$A^{+}$ is closed and transitive hence $A^{+}=K^{+}$.

Assume $(M,g)$ is chronological then by lemma \ref{vfd} $(M,g)$ is
strongly causal. The relation $A^{+}$ is antisymmetric indeed let
$(x,y)\in A^+$ and $(y,x)\in A^{+}$, $x\ne y$, and let $\sigma_n$ of
endpoints $(x_n,y_n)$ and $\gamma_n$ of endpoints $(y'_n, z_n)$ be
sequences of causal curves whose endpoints converge to the initial
pairs $(x_n,y_n) \to (x,y)$, $(y'_n, z_n)\to (y,x)$. Then we repeat
the argument used above, that is we apply the limit curve theorem to
the accumulation point $y$. Call $\sigma$ the limit causal curve for
$\sigma_n$ and analogously let $\gamma$ be the limit causal curve
for $\gamma_n$. If $\sigma$ connects $x$ to $y$ and $\gamma$
connects $y$ to $x$ then there is a closed causal curve on spacetime
a contradiction. Let $U\ni x$, $V\ni y$ be two disjoint causally
convex neighborhoods. If $\sigma$ connects $x$ to $y$ but $\gamma$
does not connect $y$ to $x$,  then it is possible to argue as above,
i.e. take $x'_k\in I^{-}(x)$, $x'_k \to x$, then for sufficiently
large $n$, which we can choose so that $n(k)>k$, $y'_{n(k)}\in
I^{+}(x'_k)\cap V $, from which it follows that there is a sequence
of causal curves of endpoints $x'_k$, $z_{n(k)}$, intersecting $V$.
But $(x'_k, z_{n(k)})\to (x,x)$ thus strong causality is violated at
$x$. The case in which $\gamma$ connects $y$ to $x$ is analogous.
The remaining case is that in which $\sigma$ is past-inextendible
and $\gamma$ is future-inextendible. Then $\gamma\circ \sigma$ is a
inextendible causal curve which by assumption is not a lightlike
line. Moreover, since strong causality holds, this curve is not
partially imprisoned in any compact, thus using the same argument as
above (i.e. taking advantage of the chronality of $\gamma\circ
\sigma$) it follows that there is a sequence of causal curves of
endpoints $x_n$, $z_n$ not all contained in a compact. Again there
is a contradiction with the strong casuality at $x$.
 %(or by the same
%argument used above to prove the transitivity of $A^+$, i.e. by
%applying the limit curve argument to the point $q$ and using the
%absence of lightlike lines one can construct a sequence of causal
%curves passing arbitrarily close to $q$ whose future and past
%endpoints converge to $p$, that is strong causality is violated at
%$p$ in contradiction with lemma \ref{vfd}.
\end{proof}

Clearly, if we could prove that $K$-causality is equivalent to
stable causality then the main theorem would follow. Unfortunately,
though there is evidence for this coincidence \cite{minguzzi07} no
proof has yet been given. In fact Seifert \cite{seifert71}, even
before the introduction of $K$-causality, gave an argument which
would have implied the equivalence. Unfortunately, he only sketched
the proof and  a recent more detailed study \cite{minguzzi07} has
shown that those arguments were inconclusive. If the two causal
properties are indeed equivalent it is probable that the proof would
be rather involved because the $K^+$ relation is not as easy to
handle as the other causal relations. Fortunately, however, it is
possible to circumvent this difficulty,  and avoid a direct proof of
the equivalence between stable causality and $K$-causality, by
working on compact stable causality. Indeed, the previous result
will be used in the following weaker form

\begin{corollary}
A chronological spacetime without lightlike lines is compactly
stably causal.
\end{corollary}

Now, the idea is to consider the property ``$(M,g)$ is compactly
stably causal and does not admit lightlike lines'' to show that it
is {\em inductive} (see lemma \ref{kgt}), that is, invariant under
enlargement of the light cones over compact sets. Then it is
possible to enlarge the light cones in a sequence of compact sets
that cover $M$ so as to obtain a causal spacetime with strictly
larger light cones (theorem \ref{vge}).

\begin{lemma}
On $(M,g)$ let $B$ be a relatively compact open set, let $g_n$ be a
sequence of metrics $g_n\ge g$, $g_n
>g$ on $B$, $g_n=g$ on $B^C$, $g_{n+1}\le g_n$, and $g_n \to g$ pointwisely on the appropriate
tensor bundle. If $(M,g)$ does not have lightlike lines then  all
but a finite number of $(M,g_n)$  do not have lightlike lines.
\end{lemma}

\begin{proof}
If not we can, passing to a subsequence, assume that all $(M,g_n)$
have lightlike lines. Denote $\gamma_n$ a respective sequence of
lightlike lines and assume there is one, say $\gamma_{\bar{n}}$,
which does not intersect $B$. Since $g_{\bar{n}}$ and $g$ coincide
outside $B$, $\gamma_{\bar{n}}$ is a $g$-causal curve. Also it is
$g$-achronal because if there are two points $p,q \in
\gamma_{\bar{n}}$ such that $(p,q) \in I^{+}_{g}$ then as $g\le
g_{\bar{n}}$, $(p,q) \in I^{+}_{g_n}$ which is impossible because
$\gamma_{\bar{n}}$ is a lightlike line on $(M,\bar{g}_n)$. But
$\gamma_{\bar{n}}$ cannot be $g$-achronal as it would be a lightlike
line of $(M,g)$, thus the overall contradiction proves that all
$\gamma_n$ intersect $B$. Without loss of generality we can assume
(pass to a subsequence if necessary) that there are $x_n \in B\cap
\gamma_n$, and $x \in \bar{B}$ such that $x_n \to x$. By the limit
curve theorem there is a inextendible $g$-causal curve $\eta$
passing through $x$. If $\eta$ is not $g$-achronal there are $y,z
\in \eta$ such that $(y,z) \in I^{+}_{g}\subset I^{+}_{g_n}$ for
every $n$.  But since $y$ and $z$ are limit points of the sequence
$\gamma_n$ and $I^{+}_g (\subset I^{+}_{g_n})$ is open some of the
curves $\gamma_n$ are not lightlike lines. The contradiction proves
that $\eta$ is not only $g$-causal but also $g$-achronal thus it is
a lightlike line. Again this is impossible thus the assumption that
an infinite number of $(M,g_n)$ does admit lightlike lines has lead
to a contradiction.

\end{proof}

\begin{lemma} \label{kgt}
If $(M,g)$ is compactly stably causal and without lightlike lines
then for every open set of compact closure $B$ it is possible to
find a metric $g_B\ge g$ such that $g_B>g$ on $B$, $g_B=g$ outside
$B$, and $(M,g_B)$ is compactly stably causal and without lightlike
lines.
\end{lemma}

\begin{proof}
Since $(M,g)$ is compactly stably causal we can find $\tilde{g}_B$
such that $\tilde{g}_B>g$ on $B$, $\tilde{g}_B=g$ outside $B$ and
$(M,\tilde{g}_B)$ is causal. Define
$g_n=(1-\frac{1}{n})g+\frac{1}{n}\tilde{g}_B$ so that $g \le g_n\le
\tilde{g}_B$ satisfies the assumptions of the previous lemma. Thus
there is a certain element of the sequence, denote it $g_B$, such
that $(M,g_B)$ does not have lightlike lines and since $g_B\le
\tilde{g}_B$, $(M,g_B)$ is causal. But every causal spacetime
without lightlike lines is compactly stably causal thus the thesis.
\end{proof}

\begin{theorem} \label{vge}
If $(M,g)$ is chronological and without lightlike lines then it is
stably causal.
\end{theorem}

\begin{proof}
Let $h$ be an auxiliary complete Riemannian metric, $x_0 \in M$, and
let $B_k= B(x_0,k)$ be the open balls of radius $k$ centered at
$x_0$. Define $g_1=g$. By the previous lemma it is possible to find
a metric $g_2>g_1$ on $B_2$, $g_2=g_1$ outside $B_2$, such that
$(M,g_2)$ is compactly stably causal and without lightlike lines.
Next repeat the argument for the relatively compact open set $B_3$
with respect to the spacetime $(M,g_2)$: there is a metric $g_3>g_2$
on $B_3$, $g_3=g_2(=g)$ outside $B_3$, such that $(M,g_3)$ is
compactly stably causal and without lightlike lines. Continue in
this way and find a sequence of metrics $g_{k+1}\ge g_k\ge g$,
$g_{k+1}>g_k$ on $B_{k+1}$. The open sets $A_1=B_2$,
$A_k=B_{k+1}\backslash\bar{B}_{k-1}$ for $k\ge 2$, cover $M$. Let
$\{\chi_k\}$ be a partition of unity so that the support of $\chi_k$
is contained in $A_k$, and define $\tilde{g}=\sum_{k=1}^{+\infty}
\chi_k g_{k+2}$ (the sum has at most two non vanishing terms at each
point) then $\tilde{g}>g$, moreover at $x\in B_k$, $\tilde{g}(x)\le
g_{k+2}(x)$, because for $n>k$, $\chi_{n}(x)=0$ (see figure
\ref{induction}). But $(M,\tilde{g})$ is causal because otherwise
there is a closed $\tilde{g}$-causal curve $\sigma$, which being a
closed set, is entirely contained in $B_s$ for some $s$. Since
$\tilde{g}\le g_{s+2}$ on $B_s$, this curve is $g_{s+2}$-causal
which contradicts the (compact stable) causality of $(M,g_{s+2})$.
Thus since $(M,\tilde{g})$ is causal and $\tilde{g}>g$, $(M,g)$ is
stably causal.

\end{proof}

\begin{figure}[ht]
\centering \psfrag{g}{ $g$} \psfrag{g1}{ $g_1$} \psfrag{g2}{ $g_2$}
\psfrag{g3}{ $g_3$} \psfrag{g4}{ $g_4$} \psfrag{g5}{ $g_5$}
\psfrag{g6}{ $g_6$} \psfrag{B1}{{$B_1$}} \psfrag{B2}{{$B_2$}}
\psfrag{B3}{{$B_3$}} \psfrag{B4}{{$B_4$}} \psfrag{B5}{{$B_5$}}
\psfrag{B6}{{$B_6$}} \psfrag{A1}{{$A_1$}} \psfrag{A2}{{$A_2$}}
\psfrag{A3}{{$A_3$}} \psfrag{A4}{{$A_4$}} \psfrag{A5}{{$A_5$}}
 \psfrag{M}{
$M$} \psfrag{gp}{$\tilde{g}$}
\includegraphics[width=8cm]{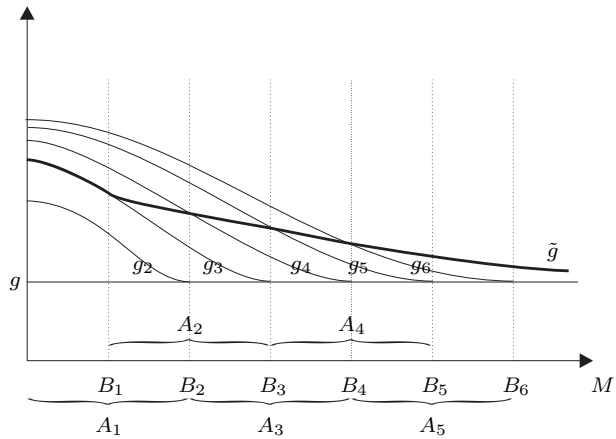}
\caption{The construction of the metric $\tilde{g}>g$ and of the
causal spacetime $(M,\tilde{g})$ in the proof of theorem \ref{vge}.}
\label{induction}
\end{figure}

\begin{remark}
This result is sharp in the sense that causal continuity can not
replace stable causality in the statement of the theorem. Indeed,
the 1+1 spacetime $\mathbb{R} \times S^1$ of coordinates
$(t,\theta)$, $\theta \in [0,2]$, metric $\dd s^2=-\dd t^2+\dd
\theta^2$ with the timelike segment $\theta=1$, $0\le t \le 1$,
removed does not have lightlike lines, is chronological, and thus
stably causal ($t$ is a time function) but it is not reflective and
hence it is not causally continuous. Analogously, chronology can not
be weakened to non-total viciousness indeed, for instance, the
spacetime of figure \ref{triangle} is non-totally vicious, does not
have lightlike lines but is not even chronological. Nevertheless, it
is possible to relax slightly the chronology condition by asking,
for instance, that the chronology violating set be confined in a
compact or even more weakly to have a compact boundary (see the next
section).
\end{remark}

Recall that a time function $t:M \to \mathbb{R}$ is a continuous
function which increases on every causal curves, that is, if
$\gamma:B \to M$ is a causal curve, $b_1<b_2$ implies
$t(\gamma(b_1))<t(\gamma(b_2))$. Hawking proved, improving previous
results by Geroch \cite{geroch70}, that stable causality holds if
and only if the spacetime admits a time function (for the direction,
time function $\Rightarrow$ stable causality, see \cite{hawking68},
for the other direction see \cite{hawking73}). Actually the time
function can be chosen smooth with timelike gradient \cite{bernal04}
(see also \cite{seifert77}). Thus a corollary of theorem \ref{vge}
is

\begin{theorem} \label{vge2}
If $(M,g)$ is chronological and without lightlike lines then it
admits a time function (which can be chosen smooth with timelike
gradient).
\end{theorem}

Recall also that if $t$ is a time function then $F_a=\{p:\, t(p)>a
\}$ is an open future set and $\dot{F}_a=\{p:\, t(p)=a \}$. In
particular, $S_a=\dot{F}_a$ is an acausal boundary (hence edgeless),
that is, $S_a$ is a {\em partial Cauchy hypersurface}
\cite{hawking73}.

The great advantage of theorem \ref{vge2}, is that it allows to
considerably weaken the causality and boundary conditions underlying
most singularity theorems. Indeed, most of them assume some of the
following: (a)  global hyperbolicity, (b) a partial Cauchy
hypersurface (c) a compact achronal edgeless set (d) a trapped set.
Often these global assumptions are made without any further
justification, in fact Senovilla in his review \cite[p.
803-8]{senovilla97} expressed the opinion that these boundary
assumptions may represent the main weak point of singularity
theorems. Fortunately, theorem \ref{vge} justifies the presence of a
foliation of partial Cauchy hypersurfaces and hence may be used to
weaken the global assumptions made in singularity theorems.

\subsection{Absence of lightlike rays}

In this section I am going to consider the implications of the
absence of lightlike rays. Recall that a future ray is a
future-inextendible causal curve which is achronal. Past rays are
defined analogously. Chosen a point $c\in (a,b)$ in a lightlike line
$\gamma: (a,b) \to M$, the portion $\gamma\vert_{[c,b)}$ is a
lightlike future ray while $\gamma\vert_{(a,c]}$  is a lightlike
past ray, thus

\begin{lemma}
The absence of lightlike future (or past) rays implies the absence
of lightlike lines.
\end{lemma}

Thus, assuming the absence of lightlike future rays one expects to
obtain a stronger property than stable causality. Indeed, we have
(see also the related result \cite[Prop. 4]{tipler77})

\begin{theorem} \label{buq}
If $(M,g)$ is chronological and without future lightlike rays then
it is globally hyperbolic (and the only TIP is $M$). An analogous
past version also holds.
\end{theorem}

\begin{proof}
%Let $\matcal{C}\ne M$ be the chronology violating set of $M$. If it
%is non-empty then there is a point $x \in \dot{\mathcal{C}}$. In the
%proof of the next theorem \ref{cfs} and argument is given which
%shows that

Since there are no future rays then there are no lightlike lines and
the spacetime is stably causal and admits a time function $t$. Let
$p\le q$, we have to prove that $C=J^{-}(q)\cap J^{+}(p)$ is
compact. Take $r \in I^{+}(q)$ so that $a=t(r)>t(q)$, and consider
the partial Cauchy surface $S_a$. Since $C\subset I^{-}(r)$, all the
points in $C$ stay in the past set $P_a=\{x:\, t(x)<a \}$. The set
$H^{-}(S_a)$ is generated by future lightlike rays (as $S_a$ is
edgeless) and since by assumption there is no future lightlike ray,
$H^{-}(S_a)$ is empty. Thus $C\subset P_a \subset D^{-}(S_a)\subset
D(S_a)$, the last set being globally hyperbolic. Note that no causal
curve from $p$ can escape $D(S_a)$ and hence $P_a$ to return to $q$,
as $t$ is a time function. Hence $C=J^{-}_{D(S_a)}(q)\cap
J^{+}_{D(S_a)}(p)$ is compact. Finally, $(M,g)$ has no TIP but $M$
because the boundary of any TIP is generated by future lightlike
rays.
\end{proof}

\begin{figure}[ht]
\centering \psfrag{M}{{\footnotesize
$\!\!\!M\backslash\bar{\mathcal{C}}$}} \psfrag{C}{{\footnotesize
$\mathcal{C}$}}  \psfrag{R}{\!\!\!\!\!\!{\footnotesize Remove}}
\psfrag{Identify}{\!\!\!\!\!\!\!\!{\footnotesize Identify}}
\includegraphics[width=7.5cm]{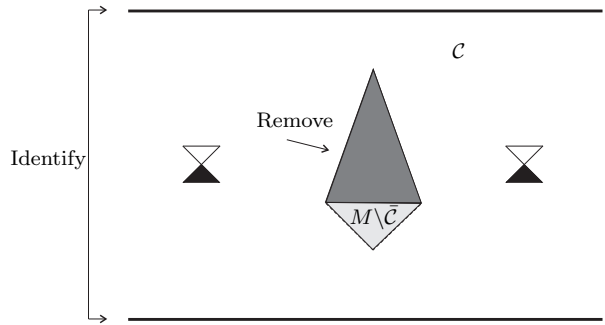}
\caption{The figure displays 1+1  Minkowski spacetime with two
spacelike slices identified and a triangle removed. If the angle at
the top of the triangle is small enough there are no  past lightlike
rays.} \label{triangle}
\end{figure}

Note that in theorem \ref{buq} chronology can not be weakened to
non-total viciousness, i.e. to the condition $\mathcal{C}\ne M$
where $\mathcal{C}$ is the chronology violating set. Indeed, figure
\ref{triangle} gives a counterexample. Nevertheless, if one replaces
the absence of {\em future} lightlike rays with the absence of
lightlike rays  then the proof of theorem \ref{cfs} will show that a
non-totally vicious spacetime is chronological (by showing that
$\dot{\mathcal{C}}$ if non-empty, contains a lightlike ray), and
thus one has:

\begin{theorem}
If $(M,g)$ is non-totally vicious and without lightlike rays then it
is globally hyperbolic (and there are no TIP or TIF but $M$).
\end{theorem}

\subsection{Physical considerations}
Theorem \ref{buq}  can be used as a singularity theorem though the
null convergence condition is not enough to guarantee that a
future-complete future-inextendible (affinely parametrized)
lightlike geodesic $\gamma: [a,+\infty) \to M$ admits a pair of
conjugate points. A sufficient condition is  Tipler's \cite[Prop.
1]{tipler77}
\begin{equation} \label{pdh}
\lim_{s \to +\infty} \,[(s-a)\int_{s}^{+\infty} \!\!\!\!\!R_{c d}
n^{c} n^{d}\, \dd s']>1,
\end{equation}
where $n^{c}$ is the tangent vector to $\gamma$ at $\gamma(s)$.
Weaker conditions were also considered by  Borde \cite{borde87}.
These conditions physically state that the energy density should not
drop off too sharply. The assumption is reasonable in those cases
where the universe is contracting (or taking the past version,
expanding) as one would expect the energy density to increase rather
than decrease.

Thus we get the following singularity theorem (past version)
\begin{theorem} \label{oioi}
The following conditions cannot all hold
\begin{itemize}
\item[(i)] $(M,g)$ is past null geodesically complete,
\item[(ii)] $(M,g)$ is chronological
\item[(iii)] $(M,g)$ is non-globally hyperbolic,
\item[(iv)]  Some energy  condition which implies the presence of conjugate points in  past-complete
past-inextendible  lightlike  geodesics (e.g. \[\lim_{s \to -\infty}
\,[(b-s)\int_{-\infty}^{s} \!\!\!\!\!R_{c d} n^{c} n^{d}\, \dd
s']>1,\] holds on any past-inextendible lightlike geodesic $\gamma:
(-\infty,b) \to M$).
\end{itemize}
\end{theorem}

The nice feature of this theorem is that there is essentially no
boundary assumption and the causality conditions are quite weak.
There is no assumption on the existence of partial Cauchy surfaces
or trapped sets. Of course, the strongest assumption which must be
physically justified is made in (iv) but the local expansion of the
Universe together with the cosmic background radiation, seem to
support it. Then the theorem states, that under the said energy
conditions the spacetime is either globally hyperbolic or has
singularities. Used in conjunction with Penrose's (1965), and
Hawking and Penrose's (1970) singularity theorems \cite{hawking73}
%and with the Lorentzian splitting theorem \cite{beem96}
it allows to
characterize quite precisely what a spacetime looks like if it
contains trapped surfaces and it is still null geodesically
complete.

%Indeed, Penrose's theorem states that the following condition for a
%spacetime cannot all hold: (i) $(M,g)$ is future geodesically
%complete, (ii) $(M,g)$ is globally hyperbolic and one (and hence
%all) space section is non-compact, (iii) There is a closed trapped
%surface, (iv) The null convergence condition.

We have

\begin{theorem} \label{pfk}
Let $(M,g)$ be a spacetime of dimension greater than 2. If
\begin{itemize}
\item[(i)] $(M,g)$ is null geodesically complete,
\item[(ii)] $(M,g)$ is chronological,
\item[(ii)] There is a future trapped surface,
\item[(iv)]  The timelike convergence,  the generic condition, together with some energy  condition which implies the presence of conjugate points in  past-complete
past-inextendible  lightlike  geodesics (e.g. \[\lim_{s \to -\infty}
\,[(b-s)\int_{-\infty}^{s} \!\!\!\!\!R_{c d} n^{c} n^{d}\, \dd
s']>1,\] holds on any past-inextendible lightlike geodesic $\gamma:
(-\infty,b) \to M$).
\end{itemize}
then the spacetime is globally hyperbolic with compact space slices
and has a incomplete timelike line.
\end{theorem}

\begin{proof}
The conditions (i), (ii) and (iv) imply (v): the spacetime is
globally hyperbolic (theorem \ref{oioi}). The Cauchy hypersurfaces
are either compact or non-compact. In the latter case (iii) and (v)
imply, by the Penrose singularity theorem, that the spacetime is
null geodesically incomplete. Thus (vi): the Cauchy hypersurfaces
are compact. The proof of  the Hawking-Penrose theorem
%, without the timelike genericity condition,
implies that (i), (ii) or (vi), and
(iv) imply that there is a incomplete timelike line.
%The line can be
%either complete or incomplete. However, if the line is complete the
%Lorentzian splitting theorem implies that the spacetime is static,
%which contradicts the existence of a trapped surface. Thus the
%timelike line is incomplete.
\end{proof}

Since the existence of trapped surfaces is a quite natural
consequence of general relativity if matter concentrate enough,
theorem \ref{buq} supports the global hyperbolicity of the spacetime
(and a closed space) provided it is null geodesically complete.
Since the conditions are quite reasonable one concludes that the
spacetime is either null geodesically incomplete or timelike
geodesically incomplete (or both).

%Here however, there is no need to assume that the strong energy
%condition (timelike convergence condition, SEC) holds. This fact is
%important, indeed SEC has often been criticized \cite{senovilla97}
%as it is the energy condition which is the less well founded
%physically: the energy density may be non-negative (WEC) and
%propagated causally (DEC) without satisfying the strong energy
%conditions; there are examples of fields, the most notable being the
%massive scalar field which do not satisfy SEC. This possibility of
%proving timelike geodesic incompleteness without using a
%``timelike''  energy condition is a consequence of the Lorentzian
%splitting theorem. A related good feature of theorem \ref{pfk} is
%that its conclusions are completing independent of the presence of a
%cosmological constant.
% We conclude that in many respects the
%previous result is stronger than the Hawking-Penrose's theorem. Its
%proof and a more thoroughly study of its physical consequences will
%be given elsewhere.

Finally I would like to stress that the assumption of null geodesic
completeness does not lead to a spacetime picture which contradicts
observations. Thus theorems \ref{buq} and \ref{vge} may have a
``positive'' role in proving the good causal property of spacetime
rather than being used only to prove its singularity. As a matter of
fact they can be used to do both (theorem \ref{pfk}).

%may be used to justify global hyperbolicity starting from weaker or
%conceptually different assumptions, it may have implications for the
%strong cosmic censorship conjecture, though at the moment this is
%more a hope that a well defined strategy.

%\begin{remark}

%\end{remark}

\section{The non-chronological case}

So far we have studied the consequence of the absence of lightlike
lines under the assumption of chronology. Let us consider the other
possibility, namely non-chronological spacetimes. Denote with
$\mathcal{C}$ the chronology violating set, with
$\mathcal{C}_\alpha$,
$\mathcal{C}=\bigcup_\alpha\mathcal{C}_\alpha$, its (open)
components and with $B_{\alpha k}$ the (closed) components of the
respective boundaries $\dot{\mathcal{C}}_\alpha= \bigcup_k B_{\alpha
k}$.

The next result joins two theorems, one by Kriele \cite[Theorem
4]{kriele89} who improved previous results by Tipler \cite{tipler77}
and the other by the author \cite{minguzzi07c}.

%Kriele proved that the sets $B_{\alpha k}$ are non-compact (end
%hence $\bar{\mathcal{C}}$ is non-compact) while I proved that the
%sets $\mathcal{C}_\alpha$ are disjoint.

\begin{theorem} \label{cfs}
A non-chronological spacetime without lightlike lines is either
totally vicious (i.e. $\mathcal{C}=M$) or it has a non-empty
chronology violating set $\mathcal{C}$, the boundaries
$\dot{\mathcal{C}}_\alpha $ of the components $\mathcal{C}_\alpha$,
 are disjoint and
the components $B_{\alpha k}$ of those boundaries are all
non-compact. In particular non-totally vicious spacetimes without
lightlike lines are non-compact.
\end{theorem}

For the proof that the sets $\dot{\mathcal{C}}_\alpha $ are disjoint
I refer the reader to \cite{minguzzi07c}.  Instead, I elaborate on
Kriele's argument
%\footnote{In my opinion Kriele's original proof
%contains a technical problem near the end where he writes ``$r \in
%\bar{I}^{+}(q')$ hence $q' \in \bar{I}^{-}(r)$'' so that
%reflectivity seems to be tacitly used there.}
by giving a slightly
different proof that the boundaries $B_{\alpha k}$ are non-compact.
Indeed, I can give a shorter proof thanks to the limit curve theorem
contained in \cite{minguzzi07c} and to the results on totally
imprisoned curves contained in \cite{minguzzi07f}.

Recall that in the chronology violating set $\mathcal{C}$, Carter's
equivalence relation $p\sim q$ iff $p \ll q\ll p$ gives rise to open
equivalence classes, moreover, since $\mathcal{C}$  is open, if $x
\in \dot{\mathcal{C}}$ it cannot be $x\in \mathcal{C}$. Recall also
that with $\Omega_f(\eta)$ it is denoted the set $\Omega_f(\eta)=
\bigcap_{t \in \mathbb{R}}\overline{\eta_{[t,+\infty)}} $ of
accumulation points in the future of the causal curve $\eta$, and
analogously in the past case. This set is always closed, moreover,
it is non-empty iff the curve is partially imprisoned in a compact
\cite{minguzzi07f}.

\begin{proof}
Assume that $B_{\alpha k} \subset \dot{\mathcal{C}}_\alpha$ is
compact and let $x \in B_{\alpha k}$.  Let $x_n \in
\mathcal{C}_\alpha$ such that $x_n \to x$, and let $U\ni x$ be a
convex set. There are closed timelike curves $\sigma_n \subset
\mathcal{C}_\alpha$ of starting and ending point $x_n$, which are
necessarily not entirely contained in $U$ (every convex set is
causal). Let $z=x$, then by the limit curve theorem
\cite{minguzzi07c} (point 2) there are two cases (corresponding to
$0<b<+\infty$, or $b=+\infty$ in that reference).

In the first case there is a closed continuous causal curve $\gamma
\in \bar{\mathcal{C}}_\alpha$ passing through $x$. It must be
achronal since if $p,q \in \gamma$, $p \ll q$, then $x\le p\ll q \le
x$ and hence $x\ll x$ which implies $x\in \mathcal{C}$ a
contradiction. Thus $\gamma$ is a geodesic with no discontinuity in
the tangent vectors at $x$. It can be extended to a lightlike line
$\gamma$ by making infinite  rounds over $\gamma$ (note that in this
case $\Omega_f(\gamma)=\Omega_p(\gamma)=\gamma$).

In the second case there are a future inextendible  continuous
causal curve $\gamma^x \subset \bar{\mathcal{C}}_\alpha$ starting at
$x$ and a past inextendible continuous causal curve $\gamma^z
\subset \bar{\mathcal{C}}_\alpha$ ending at $x$. If $\gamma^x \cap
I^+(x)\ne \emptyset$ and $\gamma^z \cap I^-(x)\ne \emptyset$ then
for sufficiently large $n$, since $I^{+}$ is open,  it would be
possible to complete a segment of $\gamma_n$ to a closed timelike
curve passing through $x$ hence $x\in \mathcal{C}$, a contradiction.
Thus $\gamma^x$ or $\gamma^z$, say $\gamma^x$, is a lightlike ray.
In particular $\gamma^x$ being a lightlike ray is achronal and hence
can not enter $\mathcal{C}_\alpha$, thus $\gamma^x \subset B_{\alpha
k}$. Now, since $ B_{\alpha k}$ is compact and $B_{\alpha k} \cap
\mathcal{C}=\emptyset$, results on totally imprisoned causal curves
can be applied \cite[theorem 3.6]{minguzzi07f}. In particular there
is a minimal non-empty closed achronal  set $\Omega \subset
\Omega_f(\gamma^x) \subset B_{\alpha k}$ such that through each
point of $\Omega$ there passes one and only one lightlike line, this
line is entirely contained in $\Omega$ and for every line $\alpha
\subset \Omega$, $\Omega_f(\alpha)=\Omega_p(\alpha)=\Omega$. Just
the existence of a lightlike line suffices to conclude the proof
that the boundaries $B_{\alpha k}$ are non-compact.

The last statement in a slightly weaker form has been first obtained
by Tipler \cite[theorem 7]{tipler77}. It follows from the
observation that a compact spacetime has a non-empty chronology
violating set $\mathcal{C}$ (see \cite[Prop. 6.4.2]{hawking73}) thus
either $\mathcal{C}=M$ or $\dot{\mathcal{C}}$ is non-empty and
compact in contradiction with the absence of lightlike lines.

\end{proof}
%
%Actually, the previous proof shows the following
%
%\begin{proposition}
%If the boundary of the chronology violating set has a compact
%component $B_{\alpha k}$, then $B_{\alpha k}$ contains a closed
%subset $\Omega$ which is a minimal non-empty closed achronal set
%such that through every point of $\Omega$ there passes one and only
%one lightlike line, this line is entirely contained in $\Omega$ and
%for every line $\alpha \subset \Omega$,
%$\Omega_f(\alpha)=\Omega_p(\alpha)=\Omega$. Moreover, any pair of
%events in $\Omega$ has the same chronological past and future.
%\end{proposition}
%

These results restrict the possible chronology violation in
spacetimes without lightlike lines, for instance they state that the
chronology violation must extend to infinity. In principle this fact
does not mean that a chronology violating region can not develop
from regular data. For this to be the case stronger global
assumptions than the only absence of lightlike lines should be
assumed \cite{tipler77,krasnikov02}.

Instead of trying to remove chronology violating sets altogether
from the spacetime, it is natural to consider what theorem \ref{vge}
may say in the cases of chronology violation. The idea is that if
$(M,g)$ has a non-empty chronology violating set but $M\ne
\bar{\mathcal{C}}$ then the spacetime $(N,g\vert_{N})$, where $N$ is
any connected components of $M\backslash \bar{\mathcal{C}}$, has
empty chronology violating set.

\begin{figure}[ht]
\centering \psfrag{N}{{\footnotesize $N$}}
\psfrag{G1}{{\footnotesize $\gamma_1$}} \psfrag{G2}{{\footnotesize
$\gamma_2$}} \psfrag{R}{\!\!\!\!\!\!{\footnotesize Remove}}
\psfrag{Identify}{\!\!\!\!\!\!\!\!{\footnotesize Identify}}
\includegraphics[width=8cm]{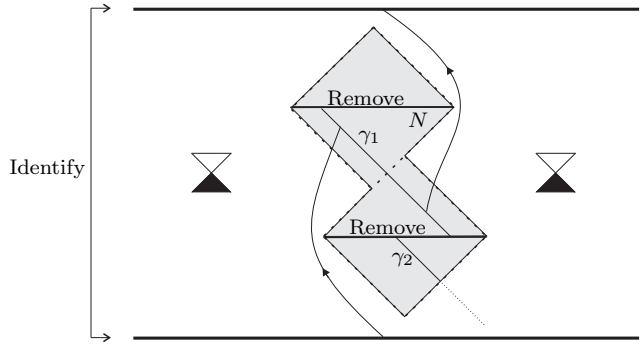}
\caption{If $(M,g)$ has a non-empty chronology violating set  and
has no lightlike line, $(N,g\vert_N)$, with $N$ any component of the
shaded region $M\backslash \bar{\mathcal{C}}$, may admit lightlike
lines (e.g. the causal curves $\gamma_1$ or $\gamma_2$).}
\label{mkl}
\end{figure}

However, even if $(M,g)$ does not have lightlike lines,
$(N,g\vert_{N})$ may have lightlike lines (see figure \ref{mkl}).
This may happen because a lightlike line $\gamma$ for
$(N,g\vert_{N})$ is not inextendible in $M$, and thus once extended
it may enter the chronology violating set  (the geodesic $\gamma_2$
in the figure). Another possibility is that while $\gamma$ is also
inextendible in $M$, the enlargement of the spacetime enlarges the
set of timelike curves and hence the possibilities that $\gamma$ is
not a line (the geodesic $\gamma_1$ in the figure). Thus it is not
possible to infer from the absence of lightlike lines for $(M,g)$
the same property for $(N,g\vert_{N})$. Actually, neither the
converse is true, the Misner spacetime  (with {region} I=$N$, see
figure 32 of \cite{hawking73}) does not have lightlike lines but its
analytic extension (I+II) where II is the chronology violating set
for I+II, does admit a lightlike line given by the Misner boundary.

There is therefore no immediate way to apply theorem \ref{vge} to
the non-chronological case apart from that of motivating on physical
grounds that some component $N$ does not have lightlike lines.

%\begin{lemma}
%If $M\ne \bar{\mathcal{C}}$ and every lightlike geodesic of $(M,g)$
%admits a pair of conjugate points then every component $N$ of
%$M\backslash\bar{\mathcal{C}}$ is such that every lightlike geodesic
%of $(N,g\vert_{N})$ admits a pair of conjugate points.
%\end{lemma}

\section{Conclusions} \label{ojb}

A proof has been given that chronological spacetimes without
lightlike lines are stably causal, and that non-totally vicious
spacetimes without lightlike rays are globally hyperbolic (together
with some other variations). The properties: (i) chronology, (ii)
null convergence condition and (iii) null generic condition, are
quite reasonable from a physical point of view, moreover, for our
purposes (ii) can be weakened to the averaged null convergence
condition.
%While deviations from these conditions are worth
%investigating we assumed them in most of the work when it came to
%draw physical conclusions from the geometric theorems.
Assuming (i), (ii) and (iii) the result of the title of this work
translates into the physical statement that  null geodesically
complete spacetimes are stably causal and therefore admit a time
function. Since the existence of some partial Cauchy surface is
assumed in most singularity theorems, this result can be used to
weaken the assumptions of those theorems.  This result may also
prove important when applied to the study of the real Universe.
Indeed, let us recall that Hawking's and Hawking and Penrose's
theorems \cite{hawking73} suggest the existence of an incomplete
causal curve which however could well be timelike. In other words
our Universe may perhaps be geodesically null complete but timelike
incomplete, in which case the main theorem could be applied in the
``positive'' way to infer the existence of a time function for the
Universe. In fact theorem \ref{pfk} shows that the assumption of
null geodesic completeness leads to consequences that do not
contradict physical observations.

The Penrose's singularity theorem seems to go against this
conclusion as it predicts null incompleteness in those cases in
which trapped surfaces form. It must be remarked, however, that
Penrose's theorem assumes the existence of a non-compact Cauchy
hypersurface thus (i) it assumes the existence of a time function
and hence it cannot be used to dismiss the conclusion that a time
function exists and (ii) for spacetimes with compact slices its
conclusions do not hold. However, even if the space slices are
compact, one can still extract information from the proof of
Penrose's theorem \cite[theorem 14.61]{oneill83}. The result is
that, roughly speaking, black holes do not exist. Trapped surfaces
may form and locally they may resemble black holes but the global
behavior would be quite different. Indeed, their horizons would
finally join and swallow the whole spacetime. Thus, without an
``exterior'', the ``interior'' could not be distinguished from a
usual spacetime.

%global hyperbolicity and hence something which is much stronger than
%the mere existence of a time function. In other words we cannot use
%it to infer the existence of null geodesic singularities, trying to
%make the results of our theorem vacuous, by assuming something which
%is even stronger that what  our theorem has been conceived to prove,
%namely the existence of a time function. Of course if the spacetime
%is assumed globally hyperbolic there is no need to prove that it is
%stably causal. Consider an hypothetical situation in which the
%spacetime is null geodesically complete, non-globally hyperbolic but
%stably causal, then no singularity theorem neither that of this work
%would seem to be plainly violated, and our theorem then could retain
%a ``positive'' value if used to prove stable causality.

In conclusion the  theorems of this work can be used physically,
either in the ``negative'' way, to prove the existence of
singularities or of chronology violating regions, or in the
``positive'' way to argue for the existence of a time function or of
global hyperbolicity. In either case they shade new light on the
existence and role of time at cosmological scales.

%In summary, the global existence of time has been reduced to weaker
%conditions which admit  straightforward physical interpretations.
%Thus  the  results of this work may shade new light

%The theorem can also be read as a singularity result as follows: if
%there is some form of causality violation on spacetime then either
%it is the worst possible, namely violation of chronology, or there
%is a singularity.

\section*{Acknowledgments}
This work has been partially supported by GNFM of INDAM and by MIUR
under project PRIN 2005 from Universit\`a di Camerino.

%\bibliography{../../bibliografie/simultaneity,../../bibliografie/libri,../../bibliografie/miei,../../bibliografie/mieiPreprints,../../bibliografie/mieiProceedings}
%\bibliographystyle{cmp}

\end{document}